\documentclass[a4paper,10pt]{article}

\usepackage[utf8x]{inputenc}
\usepackage[T1]{fontenc}
\usepackage[english]{babel}
\usepackage{amssymb}
\usepackage{amsmath}
\usepackage{amsthm}
\usepackage{fullpage}
\usepackage{verbatim}
\usepackage{enumerate}

\newtheorem{de}{Definition}[section]

\newtheorem{lem}[de]{Lemma}
\newtheorem{theo}[de]{Theorem}
\newtheorem{prop}[de]{Proposition}
\newtheorem{cor}[de]{Corollary}
\newtheorem{claim}[de]{Claim}

\newtheorem{question}[de]{Question}

\renewenvironment{proof}
{\textbf{Proof.}}
{\qed\newline}

\newcommand{\Ps}{\mathcal{P}}
\newcommand{\Ks}{\mathcal{K}}
\newcommand{\Ss}{\mathcal{S}}
\newcommand{\PS}[3]{\ensuremath{\mathbf{PS}(#1,#2,#3)}}
\newcommand{\NPS}[4]{\ensuremath{{#4}\text{-}\mathbf{NPS}(#1,#2,#3)}}

\DeclareMathOperator{\expect}{\mathbb{E}}

\newcommand{\N}{\mathbb{N}}

\newcommand{\R}{\mathbb{R}}
\newcommand{\F}{\mathbb{F}}

\newcommand{\ra}{\rightarrow}

\newcommand{\Ra}{\Rightarrow}

\title{Almost-perfect secret sharing}

\author{
Tarik Kaced\footnote{%
LIF, Univ. Aix--Marseille.
Email:~\tt{tarik.kaced@lif.univ-mrs.fr}
}
}

\begin{document}
\maketitle

\begin{abstract}
Splitting a secret $s$ between several participants, we generate (for each
value of $s$) shares for all participants. The goal: authorized groups of
participants should be able to reconstruct the secret but forbidden ones get no
information about it. In this paper we introduce several notions of
\emph{non-perfect} secret sharing, where some small information leak is
permitted. We study its relation to the Kolmogorov complexity version of secret
sharing (establishing some connection in both directions) and the effects of
changing the secret size (showing that we can decrease the size of the secret
and the information leak at the same time).
\end{abstract}


\section{Secret sharing: a reminder}

Assume that we want to share a secret -- say, a bit string $x$ of length $n$ --
between two people in such a way that they can reconstruct it together but none
of them can do this in isolation. This is simple, choose a random string $r$ of
length $n$ and give $r$ and $r\oplus x$ to the participants ($r\oplus x$ is a
bitwise \textsc{xor} of $x$ and $r$.) Both $r$ and $r\oplus x$ in isolation are
uniformly distributed among all $n$-bit strings, so they have no information
about $x$.

The general setting for secret sharing can be described as follows. We consider
some finite set $\Ks$ whose elements are called \emph{secrets}. We also have a
finite set $\Ps$ of \emph{participants}. An \emph{access structure} is a
non-empty set $\Gamma$ whose elements are groups of participants, i.e., a
non-empty subset of $2^{\Ps}$. Elements of $\Gamma$ are called
\emph{authorized} groups of participants (that should be able to reconstruct
the secret). Other subsets of $\Ps$ are called \emph{forbidden} groups (that
should get no information about the secret). We always assume that $\Gamma$ is
upward-closed (it is natural since a bigger group knows more)\footnote{One can
also consider a more general setting where some groups are neither allowed nor
forbidden (so there is no restriction on the information they may get about the
secret.) We do not consider this more general setting here.}.

In our initial example $\Ks=\mathbb{B}^n$ (the set of $n$-bit strings),
$\Ps=\{1,2\}$ (we have two participants labeled $1$ and $2$), and $\Gamma$
consists of the set $\{1,2\}$ only.

In general, \emph{perfect secret sharing} can be defined as follows. For every
participant $p\in\Ps$ a set $\Ss_p$ is fixed; its elements are $p$'s
\emph{shares}. For every $k\in\Ks$ we have a tuple of $\#\Ps$ dependent random
variables $\sigma_p\in\Ss_p$. There are two conditions:
\begin{itemize}

\item for every authorized set $A\in\Gamma$ it is possible to reconstruct
uniquely the secret $k$ from the shares given to participants in $A$ (i.e., for
different secrets $k$ and $k'$ the projections of the corresponding random
tuples onto the $A$-coordinates have disjoint ranges);

\item for every forbidden set $B\notin\Gamma$ the participants in $B$ get no
information about the secret (i.e., for different secrets $k$ and $k'$ the
projections of the corresponding random tuples onto $B$-coordinates are
identically distributed).

\end{itemize}

Various versions of combinatorial schemes were introduced in \cite{BS92} and
\cite{brickell-oclassisss}. Note that in this definition we have no
probability distribution on the set of secrets. It is natural for the setting
when somebody gives us the secret (i.e., the user chooses her password) and we
have to share whatever is given to us.

We consider another setting (as, first in \cite{Karnin-OSSS} and
further developed in \cite{capocellietal-oss}) where secret is also a random
variable. Consider a family of random variables: one ($\varkappa$)
for the secret and one ($\sigma_p$) for each participant $p$.  This family is a
perfect secret sharing scheme if 
\begin{itemize}

\item for every authorized set $A$ the projection $\sigma_A = \{\sigma_p, p\in
A\}$ determines $\varkappa$;

\item for every forbidden set $B$ the projection $\sigma_B$ is independent with
$\varkappa$.

\end{itemize}

These conditions can be rewritten using Shannon information theory: the first
condition says that $H(\varkappa|\sigma_A)=0$, and the second says that
$I(\sigma_B:\varkappa)=0$. Here $H(\cdot|\cdot)$ stands for conditional Shannon
entropy and $I(\cdot:\cdot)$ stands for mutual information. (To be exact, we
should ignore events of probability zero when saying that $\sigma_A$ determines
$\varkappa$.  To avoid these technicalities, let us agree that our probability
space is finite and all non-empty events have positive probabilities.) 

These definitions are closely related. Namely, it is easy to see that:
\begin{itemize}

\item Assume that a perfect secret sharing scheme in the sense of the first
definition is given. Then for every distribution on secrets (random variable
$\varkappa\in\Ks$) we get a scheme in the sense of the second definition as
follows. For each secret $k\in \Ks$ we have a family of dependent random
variables $\sigma_p$, and we use them as conditional distribution of
participants' shares if $\varkappa=k$.

\item Assume that a perfect secret sharing scheme in the sense of the second
definition is given, and all secrets have positive probability according to
$\varkappa$. Then the conditional distributions of $\sigma_p$ with the
condition $\varkappa=k$ form a scheme in the sense of the first definition.

\end{itemize}

This equivalence shows that in the second version of the definition the
distribution on secrets is irrelevant (as far as all element in $\Ks$ have
positive probability): we can change $\varkappa$ keeping the conditional
distributions, and still have a perfect secret sharing scheme. The advantage of
the second definition is that we can use standard techniques from Shannon
information theory (e.g., information inequalities).

The general task of secret sharing can now be described as follows: given a set
of secrets $\Ks$ and an access structure $\Gamma$ construct a secret sharing
scheme. This is always possible (see~\cite{benaloh-gssmf,ito-generalas}).
However, the problem becomes much more difficult if we limit the size of
shares. It is known (see~\cite{capocellietal-oss}) that in the non-degenerate
case shares should be at least of the same size as the secret: $\#\Ss_p\ge
\#\Ks$ for every essential participant $p$. (A participant is \emph{essential}
if we remove it from some authorized group and get a forbidden group.
Evidently, non-essential participants can be just ignored.) This motivates the
notion of \emph{ideal} secret sharing scheme where $\#\Ss_p=\#K$ for every
essential participant $p$.

Historically, the motivating example for secret sharing was Shamir's scheme
(see~\cite{shamir-howtoss}). It has $n$ participants, authorized groups are
groups of $t$ or more participants (where $t$ is an arbitrary threshold).
Secrets are elements of a finite field $\F$ of size greater than $n$. To share
a secret $k$, we construct a polynomial $$
P_k(x)=k+r_1x+r_2x^2+\ldots+r_{t-1}x^{t-1} $$ where the $r_i$ are chosen
independently and uniformly. The shares are the values $P(x_1),\ldots,P(x_{n})$
for distinct nonzero field elements $x_1,\ldots,x_{n}$ (for each participant a
non-zero element of the field is fixed). Any $t$ participants together can
reconstruct the polynomial while for any $t-1$ participants all combinations of
shares are equally probable (for every $k$). This scheme is ideal.

Not every access structure allows an ideal secret sharing scheme. For example,
no ideal scheme exists for four participants $a,b,c,d$ where the authorized
groups are $\{a,b\}$, $\{b,c\}$ and $\{c,d\}$ and all their supersets
(see~\cite{benaloh-gssmf,kurosawa-clb}; it is shown there that every secret
sharing scheme for this access structure satisfies $\log \#S_b+\log \#S_c\ge
3\log \#\Ks$).

It is therefore natural to weaken the requirements a bit and to allow non-ideal
secret sharing schemes still having shares of reasonable size. For example, we
may fix some $\rho\ge 1$ and ask whether for a given access structure there
exists a perfect secret sharing scheme where $\max_{p\in\Ps}\log\#S_p\le
\rho\log\#\Ks$.  (The answer may depend on the size of $\Ks$.)


Unfortunately, not much is known about this. There are quite intricate lower
bounds for different specific access structures (some proofs are based on
non-Shannon inequalities for entropies of tuples of random variables, see
\cite{beimel-ssnsii,metcalfb-iub}).  The best known lower bounds
for sharing $m$-bit secrets (for some fixed access scheme) are still rather
weak, like $\frac{n}{\log n}m$ (see \cite{csirmaz-ssmustbelarge}). On the other
hand, the known upper bounds for general access structures are exponential in
the number of participants (and rather simple, see
\cite{benaloh-gssmf,ito-generalas}).


\section{Nonperfect secret sharing}

The relaxation of the perfectness property is natural when efficiency is
involved (see \cite{beimel-wpss,kurosawa-npsssm,srinathan-npssgas}). Our
attempt here is to encapsulate existing definitions of non-perfect schemes in
the Shannon framework. We consider possible relaxations of the requirements and
introduce several versions of \emph{almost-perfect} secret sharing. By this we
mean that we allow limited ``leaks'' of information to forbidden groups of
participants. We also consider schemes where authorized groups need some
(small) additional information to reconstruct the secret.  Such
approximately-perfect schemes are quite natural from the practical point of
view.  Also, the gain in flexibility may help overcome the difficulty of
constructing efficient perfect schemes which seems related to difficult
problems of combinatorial or algebraic nature. 

Let us discuss possible definitions for almost-perfect schemes. Now we want to
measure the leak of information (or the amount of missing information), and the
most natural way is to replace the equations $H(\varkappa|\sigma_A)=0$ and
$I(\sigma_B:\varkappa)=0$ by inequalities $H(\ldots)<\varepsilon_1$ and
$I(\ldots)<\varepsilon_2$, for some bounds $\varepsilon_1$ and $\varepsilon_2$
(normally, a small fraction of the amount of information in the secret itself).

The problem here is that measuring the information leak and missing information
in this way, we need to fix some distribution on secrets, and this looks
unavoidable even from the intuitive point of view. Imagine that we have
$1000$-bit secrets, and the sharing scheme works badly for secrets with $900$
trailing zeros (e.g., discloses them to all participants).  If the information
leak might not be huge for the uniform distribution, since $100$ leaked bits
are multiplied by $2^{-900}$ probability to have $900$ trailing zeros; it can
however become significant if the secret is not chosen uniformly, e.g. the user
chooses a short password padded with trailing zeros.

An interesting question (that we postpone for now) is how significant could be
this dependence. One may expect that a good secret sharing scheme remains
almost as good if we change slightly the distribution, but we cannot prove any
natural statement of this kind. So we have to include the distribution on
secrets in all the definitions.

Let $\Gamma$ be an access structure. Let $\varkappa$ and $\sigma_p$ (for all
participants $p$) be some random variables (on the same probability space, so
we may consider their joint distribution). Such a family is called a (not
necessarily perfect) secret sharing scheme, and its parameters are:
\begin{itemize}

\item distribution on secrets (in particular, the entropy of $\varkappa$ is
important);

\item \emph{information rate}, $H(\varkappa)$, the entropy of the secret
divided by the maximal entropy of a single share;

\item \emph{missing information ratio}, the maximal value of
$H(\varkappa|\sigma_A)$ for all authorized $A$, divided by $H(\varkappa)$;

\item \emph{information leak ratio}, the maximal value of
$I(\sigma_B:\varkappa)$ for all forbidden $B$, divided by $H(\varkappa)$.

\end{itemize}

To simplify our statements, we consider asymptotic behaviors and give the
following template definition of almost-perfect secret sharing:

\begin{de}
An access structure $\Gamma$ on the set $P$ of participants can be
almost-perfectly implemented with parameters
$(\rho,\varepsilon_1,\varepsilon_2)$ if there exists a sequence of secret
sharing schemes for the secret variable $\varkappa_n$, such that 
\begin{itemize}
\item $H(\varkappa_n)\to\infty$;
\item the $\limsup$ of the information rates does not exceed $\rho$;
\item the missing information ratio converges to $\varepsilon_1$ as $n\to\infty$;
\item the information leak ratio converges to $\varepsilon_2$ as $n\to\infty$.
\end{itemize}
\end{de}


In this article we introduce several definitions of almost-perfect secret
sharing schemes. Two versions in the framework of Shannon entropy for which
we show that the stronger definition, where we require no missing information,
gives the same notion; one version in the framework of Kolmogorov complexity.
We prove that all these approaches are asymptotically equivalent (have
equivalent asymptotical rates of schemes for each access structure). Hence,
we can combine tools of Shannon's information theory and Kolmogorov
complexity to investigate the properties of nonperfect secret sharing schemes.

Rather than providing constructions or stating trivial counterparts of known
theorems, we emphasize our study on the behaviour of such schemes. Simple
properties of perfect schemes provide new natural questions for nonperfect
schemes which are in general not trivial. The main contribution of the paper is
the proof of few of such natural properties, namely and Proposition
\ref{prop-transform} and Theorem \ref{maintheo} for scaling down a nonperfect
scheme while keeping roughly the same information leak ratio.

We believe our modest contribution is a small step towards a promising path to
discover new constructions and theorems in nonperfect secret sharing.

\subsection{Definitions}

We consider two different versions of the definition of approximately-perfect
secret sharing schemes. In the first one, non-perfect secret sharing schemes
are allowed to give some information to forbidden groups and/or not give
authorized groups the entire secret: 

\begin{de}
\label{def-approx-e1-e2}
Let $\Ks$ be a finite set of secrets, a $(\varepsilon_1,\varepsilon_2)$-nonperfect
secret sharing scheme for secrets in $\Ks$ implementing an access structure
$\Gamma$ is a tuple of jointly distributed discrete random variables
$(\varkappa,\sigma_1,\ldots,\sigma_n)$ such that
\begin{itemize}
\item if $A\in\Gamma$ then $H(\varkappa|\sigma_A)\le\varepsilon_1H(\varkappa)$
\item if $B\notin\Gamma$ then $I(\varkappa:\sigma_B)\le\varepsilon_2H(\varkappa)$
\end{itemize}
\end{de}

In this definition, authorized groups may fail to recover at most $\varepsilon_1$
bits of the secret while forbidden groups can not learn more than $\varepsilon_2$
bits. A probably more natural version of a non-perfect scheme is asymmetric:
authorized groups know everything about the secret, while forbidden groups can
keep not more than $\varepsilon$ bits of information about the secret:

\begin{de}
\label{def-approx-e}
Let $\Ks$ be a finite set of secrets, a $\varepsilon$-nonperfect secret sharing
scheme for secrets in $\Ks$ implementing an access structure $\Gamma$ is a
tuple of jointly distributed discrete random variables
$(\varkappa,\sigma_1,\ldots,\sigma_n)$ such that
\begin{itemize}
\item if $A\in\Gamma$ then $H(\varkappa|\sigma_A)=0$
\item if $B\notin\Gamma$ then $I(\varkappa:\sigma_B)\le\varepsilon H(\varkappa)$
\end{itemize}
\end{de}

By $\NPS{\Gamma}{N}{S}{\varepsilon}$, resp.
$\NPS{\Gamma}{N}{S}{(\varepsilon_1,\varepsilon_2)}$, we refer to a
$\varepsilon$-nonperfect, resp. $(\varepsilon_1,\varepsilon_2)$-nonperfect,
secret sharing scheme implementing access structure $\Gamma$ for $N$-bit
secrets with single shares of entropy at most $S$. We use $\PS{\Gamma}{N}{S}$
for perfect schemes, i.e., when it is the case that $\varepsilon_1$ and
$\varepsilon_2$ are null.

We now introduce the \emph{almost-perfect} versions of secret sharing, that
denotes an asymptotic sequence of nonperfect schemes for a fixed access
structure where the leak can be made negligible as the size of the secret
grows.

\begin{de}
We say that an access structure $\Gamma$ can be almost-perfectly implemented,
with parameters $(\rho,\varepsilon_1,\varepsilon_2)$, if there exists a sequence
of nonperfect schemes in the sense of Definition~\ref{def-approx-e1-e2} such
that parameters converge to $(\rho,\varepsilon_1,\varepsilon_2)$.  i.e., if
$$
\exists (\NPS{\Gamma}{N_m}{S_m}{(\varepsilon^1_m,\varepsilon^2_m)})_{m\in\N} 
\text{ s.t. }
(\varepsilon^1_m,\varepsilon^2_m)\to(\varepsilon_1,\varepsilon_2) \text{ and } 
{N_m\over S_m}\to \rho
 \text{ as } m\to\infty
$$
Moreover, we say that $\Gamma$ can be almost-perfectly implemented without
missing information when the nonperfect schemes are in the sense
of Definition~\ref{def-approx-e}.
\end{de}

\begin{prop}
Let $\Gamma$ be an access structure and $\rho$ be a positive real, the
following are equivalent
\begin{itemize}

\item $\Gamma$ can be almost-perfectly implemented 

\item $\Gamma$ can be almost-perfectly implemented without
missing information.

\end{itemize}
\end{prop}

This proposition is a corollary of the following result: one can transform a
scheme with some missing information into a scheme without missing information
by increasing the size of shares.

The natural idea to prove this is to add the missing information to authorized
groups. However this is already not trivial to implement. Indeed, we want to
keep the leak small, hence we can not use a perfect scheme to share the missing
information.  The plan is to "materialize" the missing information and add it
to each participant. The small amount of information will therefore also
increase the information leak by a small amount.  The proposition tells us that
we can indeed achieve a new leak comparable to the previous one.

\begin{prop}
\label{prop-transform}
If $\Gamma$ is an access structure on $n$ participants, then
$$
\exists\NPS{\Gamma}{N}{S}{(\varepsilon_1,\varepsilon_2)} \Ra 
\exists\NPS{\Gamma}{N}{S+O(\varepsilon_1 N2^n)}
                   {(\varepsilon_2+O(\varepsilon_1N2^n))}
$$
\end{prop}
\begin{proof}
Assume there is a $\NPS{\Gamma}{N}{S}{(\varepsilon_1,\varepsilon_2)}$, let us
transform it as follows. Take a minimal authorized set $A\in\Gamma^-$, by
definition it holds that $H(\varkappa|\sigma_A)\le\varepsilon_1N$. Informally,
it means that $A$ lacks $\varepsilon_1N$ bits of information about the secret.
We materialize this information and add it to $A$. More precisely, we use
the following lemma about conditional descriptions:

\begin{lem}\label{cond-lemma}
Let $\alpha$ and $\beta$ be two random variables defined on the same space.
Then there exists a variable $\gamma$ (defined on the same space) such that
$H(\alpha|\beta,\gamma)=0$ and $H(\gamma)\le 2H(\alpha|\beta)+O(1)$.
\end{lem}

\begin{proof}
Let $\beta$ be distributed on a set $\{b_1,\ldots, b_s\}$. For each fixed value
$b_j$, we have a conditional distribution on values of $\alpha$ given the
condition $\beta=b_j$.  We can construct for this conditional distribution on
values of $\alpha$ a prefix-free binary code $c_{1j},\ldots,c_{mj}$ such that
the average length of codewords is at most $H(\alpha | \beta=b_j)+1$ (e.g., we
can take Huffman's code).

Let $\gamma$ be the corresponding codeword: if $\beta = b_j$ and $\alpha=a_i$
then $\gamma=c_{ij}$ (the $i$-th codeword from the code constructed for the
distribution of $\alpha$ under condition $\beta = b_j$). 

Given a value $b_j$ of $\beta$ and a codeword $c_{ij}$ from the corresponding
code, we can uniquely determine the corresponding value of $\alpha$. Hence, we
get $H(\alpha | \beta, \gamma)=0$. It remains to estimate entropy of $\gamma$.

The defined above $\gamma$ ranges over the union of all codewords $c_{ij}$
(from all codes constructed for all possible values of $\beta$). The average
length of bit strings $c_{ij}$
 $$
 \expect\limits_{ij} |c_{ij}|  = \expect\limits_i (\expect\limits_{j} |c_{ij}| ) <  \expect\limits_i (H(\alpha|\beta=b_j)+1 ) = 
 H(\alpha|\beta=b_j)+1.
 $$
This observation is enough to estimate the entropy of $\gamma$.
 
The union of all codewords $c_{ij}$ is not necessarily prefix-free even if the
codes $\{c_{1j},\ldots,c_{mj}\}$ were prefix-free for each value of $\beta$.
However, we can convert any set of bit strings into a prefix-free code by a
simple transformation: we double each bit in each string, and add at the end of
each string the pair of bits $01$. E.g., a string $00101$ is converted into
$000011001101$. This simple trick converts the set of $c_{ij}$ into a
prefix-free set $c'_{ij}$ such that

$$
  \expect\limits_{ij} |c'_{ij}| = 2 \expect\limits_{ij} |c_{ij}| + 2
$$

Thus, random variable $\gamma$ can be considered as a distribution on this
prefix-free set $c'_{ij}$. It is well known that for any distribution on a
prefix-free set, the entropy is not greater than the average length of
codewords (it follows from Kraft's inequality). Hence, entropy of $\gamma$ is
not greater than the average length of $c'_{ij}$, i.e., not greater than
$2H(\alpha|\beta) + O(1)$.  
\end{proof}

We apply lemma~\ref{cond-lemma} to encode the secret $k$ conditional to the
shares of $A$. Since this random variable has entropy at most $\varepsilon_1N$,
the encoding can be done by strings of size at most $O(\varepsilon_1N )+ O(1)$.
We add this ``conditional description'' to any participant of $A$. Now the
participants of $A$ can together determine the secret uniquely.  We do the same
for all minimal authorized groups in $\Gamma^-$. So, now all authorized groups
have all information about the secret.

We added some additional data to several participants (some participants can
obtain several different ``conditional descriptions'' since one participant can
belong to several minimal authorized groups). However all additional
information given to participants is of size only  $O(\varepsilon_1N2^n)$,
hence, the extra information is given to forbidden groups is at most
$O(\varepsilon_1N2^n)$. The size of the shares in the new schemes is at most
$S+O(\varepsilon_1N2^n)$, and we are done.  
\end{proof}

An interesting open question about almost-perfect secret sharing is to settle
whether it is equivalent to perfect secret sharing or not:

\begin{question} 
Can we achieve essentially better information rates with almost-perfect schemes
than with perfect schemes ?
\end{question}

A weaker form of this question where leaks are exactly zero has been answered
by Beimel~et~al in \cite{beimel-mfariss} (using a result of Mat\'{u}\v{s}
\cite{MatusMatrRepr}) where they construct a \emph{nearly-ideal} access
structure, i.e. access structure that can be implemented perfectly with an
information rate as close to $1$ as we want but not equal. In fact, with the
same kind of arguments we can construct an almost-perfect scheme for the same
access structure with small leaks but information rate exactly one.

\begin{prop}
There is an access structure which can be implemented by an almost-perfect
scheme with parameters $(1,0,0)$ and rate exactly one but has no ideal perfect
scheme.
\end{prop}
\begin{proof}
An access structure $\Gamma$ is induced by a matroid
$M=(\mathcal{Q},\mathcal{C})$ through $s\in\mathcal{Q}$ if $\Gamma$ is defined
on the set of participants $\mathcal{P}=\mathcal{Q}\setminus\{s\}$ by the upper
closure of the collection of subsets $A\subseteq\mathcal{P}$  such that
$A\cup\{s\} \in \mathcal{C}$ (here $\mathcal{C}$ is the set of circuits of the
matroid $\mathcal{M}$.) Let $\mathcal{F}$ and $\mathcal{F}^-$ be respectively
the access structures induced by the Fano and by the non-Fano matroids (through
any point). In \cite{MatusMatrRepr}, Mat\'{u}\v{s} proved that there exist
perfect ideal schemes for $\mathcal{F}$, resp. $\mathcal{F}^-$ if and only if
$\#\mathcal{K}$ is even, resp. odd.

Consider an access structure $\Gamma$ consisting of disjoint copies of
$\mathcal{F}$ and $\mathcal{F}^-$. From Mat\'{u}\v{s} argument, $\Gamma$
cannot be implemented ideally by a perfect scheme. Construct a scheme $\Sigma$
consisting of the concatenation of two independent schemes:
\begin{itemize}

\item a $\PS{\mathcal{F}}{N}{N}$, and

\item a $\PS{\mathcal{F}^-}{N}{M}$, constructed from a
$\PS{\mathcal{F}^-}{M}{M}$ for $\#\mathcal{K}= 2^N+1$ (i.e., $M=\log(2^N+1)$)
where we removed one possible value of the secret. 

\end{itemize}

$\Sigma$ is a perfect scheme for $\Gamma$ with rate
$\frac{N}{\log(2^{N}+1)}$. Now instead of using a $\PS{\mathcal{F}^-}{N}{M}$
as second scheme, we modify it into a nonperfect scheme by substituting the
value of the share "$2^N+1$" by any other possible value. Now there are exactly
$2^N$ shares. It is not difficult to show that $\Sigma'$ is, at most, a
$\NPS{\Gamma}{N}{N}{(\frac3{N},0)}$ i.e., with information rate exactly
one.
\end{proof}

\section{Kolmogorov secret sharing}
\label{sec:KC}

We denote "the" Kolmogorov complexity function by the letter $K$. Since most
variants are equal up to a logarithmic term and our results are asymptotic. For
a complete introduction to Kolmogorov complexity and to some techniques used
here, we refer the reader to the book \cite{livitanyi} and to \cite{ZL}. 

The problem of secret sharing could be studied also in the framework of the
algorithmic information theory. The idea is that now a secret sharing scheme is
not a distribution on binary strings  but an individual tuple of binary strings
with corresponding properties of ``secrecy''. To define these ``secrecy''
properties for individual strings, we substitute Shannon's entropy by
Kolmogorov complexity and get algorithmic counterparts of the definition of
secret sharing schemes. A similar idea was realized in Definition~21~(part 1)
in \cite{ALPS} for a special case (for threshold access structures).
 
For Kolmogorov complexity there is no natural way to define an "absolutely"
perfect version of secret sharing scheme. Thus, in the framework of Kolmogorov
complexity we can deal only with ``approximately-perfect'' versions of the
definition. We define approximately-perfect secret sharing schemes for
Kolmogorov complexity just in the same way as we defined
$(\varepsilon_1,\varepsilon_2)$-nonperfect schemes for Shannon's entropy
(similarly to Definition~\ref{def-approx-e1-e2}):

\begin{de}
\label{def-kolmogorov-sss}
For an access structure $\Gamma$ we say that a tuple of binary strings
$(s,a_1,\ldots,a_n)$ is a Kolmogorov $(\varepsilon_1,\varepsilon_2)$-perfect
secret sharing scheme for secrets of size $N$ if
\begin{itemize}
\item $K(s) = N$
\item for $A\in\Gamma, K(s|a_A)\le\varepsilon_1N$
\item for $B\notin\Gamma, K(s)-K(s|a_B) = I(s:a_B)\le\varepsilon_2N$ 
\end{itemize}
\end{de}

We reuse the template of almost-perfect secret sharing, this time in the
Kolmogorov setting using the above version of secret sharing scheme. Thus, it
should make sense to talk about almost-perfect secret sharing in the sense of
Kolmogorov. 

It turns out that problems of constructing approximately perfect secret sharing
schemes in Shannon's and Kolmogorov's frameworks are closely related. For every
access structure, in both frameworks the asymptotically optimal rates  are
equal to each other.  More precisely, we have the following equivalence:

\begin{theo} \label{th-sh-k}
Let $\Gamma$ be an access structure over $n$ participants and $\rho$ be a 
positive real, then the following are equivalent:
\begin{itemize}

\item $\Gamma$ can be almost-perfectly implemented with parameters $(\rho,\varepsilon_1,\varepsilon_2)$ in the
sense of Shannon.

\item $\Gamma$ can be almost-perfectly implemented with parameters $(\rho,\varepsilon_1,\varepsilon_2)$ in the
sense of Kolmogorov.

\end{itemize}
\end{theo}

This theorem follows from a more general parallelism between Shannon entropy
and Kolmogorov complexity. Below we explain this parallelism in terms of
realizable complexity and entropy profiles.

The Kolmogorov complexity profile of a tuple $[a] = (a_1,\ldots,a_n)$ of a
binary string is defined by the vector $\vec{K}([a])$ of Kolmogorov
complexities of all pairs, triples \ldots of strings $a_i$. So, it
consists consists of $2^n-1$ (integer) complexity values, one for each
non-empty subset of $n$ strings $a_i$. In the same way we define the entropy
profile $\vec{H}([s])$ of a tuple $[s]=(s_1,\ldots,s_n)$ of random variables by
replacing $K(\cdot)$ by $H(\cdot)$.

Next theorem explains that the class of realizable complexity profiles and 
the class of entropy profiles are in some sense very similar:

\begin{theo}\label{th-eq}
For every $\vec{v}\in \R_+^{2^n-1}$ the following conditions are equivalent:
\begin{itemize}
\item there is a sequence $([s_m])_{m\in\N}$ of $n$-tuple of random variables
s.t. $\frac1{m}\vec{H}([s_m])\to \vec{v}$
\item there is a sequence $([a_m])_{m\in\N}$ of $n$-tuple of binary strings
s.t. $\frac1{m}\vec{K}([a_m])\to \vec{v}$
\end{itemize}  
\end{theo}
\noindent
Note that Theorem~\ref{th-sh-k} follows immediately from Theorem~\ref{th-eq}.

We denote "the" Kolmogorov complexity function by the letter $K$. Since most
variants are equal up to a logarithmic term and our results are asymptotic. For
a complete introduction to Kolmogorov complexity and to some techniques used
here, we refer the reader to the book \cite{livitanyi} and to \cite{ZL}. 

The problem of secret sharing could be studied also in the framework of the
algorithmic information theory. The idea is that now a secret sharing scheme is
not a distribution on binary strings  but an individual tuple of binary strings
with corresponding properties of ``secrecy''. To define these ``secrecy''
properties for individual strings, we substitute Shannon's entropy by
Kolmogorov complexity and get algorithmic counterparts of the definition of
secret sharing schemes. A similar idea was realized in Definition~21~(part 1)
in \cite{ALPS} for a special case (for threshold access structures).
 
For Kolmogorov complexity there is no natural way to define an "absolutely"
perfect version of secret sharing scheme. Thus, in the framework of Kolmogorov
complexity we can deal only with ``approximately-perfect'' versions of the
definition. We define approximately-perfect secret sharing schemes for
Kolmogorov complexity just in the same way as we defined
$(\varepsilon_1,\varepsilon_2)$-nonperfect schemes for Shannon's entropy
(similarly to Definition~\ref{def-approx-e1-e2}):

\begin{de}
\label{def-kolmogorov-sss}
For an access structure $\Gamma$ we say that a tuple of binary strings
$(s,a_1,\ldots,a_n)$ is a Kolmogorov $(\varepsilon_1,\varepsilon_2)$-perfect
secret sharing scheme for secrets of size $N$ if
\begin{itemize}
\item $K(s) = N$
\item for $A\in\Gamma, K(s|a_A)\le\varepsilon_1N$
\item for $B\notin\Gamma, K(s)-K(s|a_B) = I(s:a_B)\le\varepsilon_2N$ 
\end{itemize}
\end{de}

We reuse the template of almost-perfect secret sharing, this time in the
Kolmogorov setting using the above version of secret sharing scheme. Thus, it
should make sense to talk about almost-perfect secret sharing in the sense of
Kolmogorov. 

It turns out that problems of constructing approximately perfect secret sharing
schemes in Shannon's and Kolmogorov's frameworks are closely related. For every
access structure, in both frameworks the asymptotically optimal rates  are
equal to each other.  More precisely, we have the following equivalence:

\begin{theo} \label{th-sh-k}
Let $\Gamma$ be an access structure over $n$ participants and $\rho$ be a 
positive real, then the following are equivalent:
\begin{itemize}

\item $\Gamma$ can be almost-perfectly implemented with parameters $(\rho,\varepsilon_1,\varepsilon_2)$ in the
sense of Shannon.

\item $\Gamma$ can be almost-perfectly implemented with parameters $(\rho,\varepsilon_1,\varepsilon_2)$ in the
sense of Kolmogorov.

\end{itemize}
\end{theo}

This theorem follows from a more general parallelism between Shannon entropy
and Kolmogorov complexity. Below we explain this parallelism in terms of
realizable complexity and entropy profiles.

The Kolmogorov complexity profile of a tuple $[a] = (a_1,\ldots,a_n)$ of a
binary string is defined by the vector $\vec{K}([a])$ of Kolmogorov
complexities of all pairs, triples \ldots of strings $a_i$. So, it
consists consists of $2^n-1$ (integer) complexity values, one for each
non-empty subset of $n$ strings $a_i$. In the same way we define the entropy
profile $\vec{H}([s])$ of a tuple $[s]=(s_1,\ldots,s_n)$ of random variables by
replacing $K(\cdot)$ by $H(\cdot)$.

Next theorem explains that the class of realizable complexity profiles and 
the class of entropy profiles are in some sense very similar:

\begin{theo}\label{th-eq}
For every $\vec{v}\in \R_+^{2^n-1}$ the following conditions are equivalent:
\begin{itemize}
\item there is a sequence $([s_m])_{m\in\N}$ of $n$-tuple of random variables
s.t. $\frac1{m}\vec{H}([s_m])\to \vec{v}$
\item there is a sequence $([a_m])_{m\in\N}$ of $n$-tuple of binary strings
s.t. $\frac1{m}\vec{K}([a_m])\to \vec{v}$
\end{itemize}  
\end{theo}
\noindent
Note that Theorem~\ref{th-sh-k} follows immediately from Theorem~\ref{th-eq}.

\begin{proof} 
To prove this result, we convert a sequence of $n$-tuple of random variables
into a sequence of $n$-tuple of binary strings and visa-versa; these
conversions will preserve complexity/entropy profiles: corresponding tuples of
random variables and  strings will have similar values in their profiles.

The main technical tools are the Kolmogorov--Levin theorem
 $$K(a,b) = K(a) + K(b|a) + O(\log|ab|)$$
and  the ``typization'' trick for entropy and Kolmogorov complexity
(the same technique as in \cite{HRSV,RomashMIPair}).  

[Kolmogorov $\ra$ Shannon]
Let $[a] = (a_1,\ldots,a_n)$ be an $n$-tuple of binary strings.
For a non-negative integer $c$ (to be fixed below) we consider the following set: 
$$
T_c([a])=\left\{[a']=(a_1',\ldots,a'_n):\forall U\subseteq[1,\ldots,n],
K(a_U)-c\log|a|\le K(a'_U)\le K(a_U)\right\},
$$
which is the set of $n$-tuples of binary strings whose complexity profile is
close to the one of $[a]$ up to a logarithmic term. Further we formulate
several properties of  $T_c([a])$.

\begin{claim}\label{claim-1}
$\log \#T_c([a]) = 2^{K(a)-O(\log K(a))}$ for all large enough $c$.
\end{claim}
\begin{proof}
See Lemma~2 in~\cite{HRSV} and Proposition~1 in~\cite{RomashMIPair}.  We fix
value $c$ so that Claim~\ref{claim-1} holds ($c$ depends on the size $n$ of the
tuple but not on $K(a)$).
\end{proof}

\begin{claim}
\label{claim-conditionalK}
$\forall a'\in T_c(a),\forall U,V\subseteq[1,\ldots,n], K(a'_U|a'_V) =
K(a_U|a_V)-O(\log|a|) $
\end{claim}
\begin{proof}
Follows from the definition of $T_c(a)$ and the Kolmogorov--Levin theorem.
\end{proof}

Now, define $[s]=(s_1,\ldots,s_n)$ as an $n$-tuple of random variables
uniformly distributed on $T_c([a])$.  From the definition of $[s]$ and
Claim~\ref{claim-1} it follows that entropy of all $[s]$ is close to $K(a)$. We
claim that in fact all components of the entropy profile of $[s]$ are close to
the corresponding components in the complexity profile of $[a]$. We prove this
property in two steps. At first, we obtain /Can upper bound:

\begin{claim}
$\forall U\subseteq[1,\ldots,n], H(s_U) \le K(a_U) +1$ 
\end{claim}
\begin{proof}
The number of possible values for $s_U$ is the number of possible substrings
$a'_U$ for $a'\in T(a)$. Since $K(a'_U)\le K(a_U)$, there is at most
$2^{K(a_U)+1}-1$ such values for $s_U$. Shannon's entropy of a random variable
cannot be greater than logarithm of the number of its values, and we are done.
\end{proof}

Further, we prove the lower bound:
\begin{claim}
$\forall U\subseteq[1,\ldots,n], H(s_U) \ge K(a_U) - O(\log|a|)$
\end{claim}
\begin{proof}
First, consider $a'_U$ for some fixed $a'\in T(a)$. From Claim
\ref{claim-conditionalK},
$K(a'_{\overline{U}}|a'_U) \le K(a_{\overline{U}}|a_U) + O(\log|a|)$, thus
$s_U$ can take at most $2^{K(a_{\overline{U}}|a_U)+O(\log|a|)}$
values. This is true for all such $a'_U$, therefore
$H(s_{\overline{U}}|s_U) \le
K(a_{\overline{U}}|a_U)+O(\log|a|)$.

Then,
$$
\begin{array}{rclr}
H(\mathbf{s}_U) &=& H(\mathbf{s}) - H(\mathbf{s}_{\overline{U}}|\mathbf{s}_U)
& \text{(equality for entropy)}\\
                &\ge& K(a) - K(a_{\overline{U}}|a_U) - O(\log|a|) 
& \text{(by definition of $s$)}\\
                &\ge& K(a_U) - O(\log|a|) 
& \text{(from symmetry of information)}
\end{array}
$$
\end{proof}

Therefore, the random variable $[s]$ has an entropy profile close to the
complexity profile of $[a]$ up to a logarithmic factor. The first part for the
theorem is proven.

\medskip

[Shannon $\ra$ Kolmogorov]
Let $s = (s_1,\ldots,s_n)$ be a $n$-tuple of random variables.  We fix an
integer $M>0$ (to be specified below) and construct some $M\times n$ table
$$
\begin{array}{c}
a_1^1a_2^2 \ldots a_1^M\\
a_2^1a_2^2 \ldots a_2^M\\
\vdots\\
a_n^1a_n^2 \ldots a_n^M\\
\end{array}
$$
satisfying the following properties:
\begin{enumerate}[(a)]

\item The columns of the table  (each column is an $n$-vector) consist of
possible values for the random variable $[s]$.

\item Different $n$-tuples are used as columns in the matrix with different
frequencies; we require that each frequency is close to the corresponding
probability in the distribution of $[s]$. More precisely, for every $n$-tuple
of letters $(\alpha_1,\ldots,\alpha_n)$
$$
\mbox{the column }
\left(
\begin{array}{c}
\alpha_1\\
\alpha_2\\
\vdots\\
\alpha_n\\
\end{array}
\right)
\mbox{ should occur in the table }\ \mathbf{Prob}[s = (\alpha_1,\ldots,\alpha_n)] \cdot M + O(1)
\mbox{ times}.
$$

\item The table has the maximal Kolmogorov complexity among all tables
satisfying (a) and (b). It implies, by a rather simple counting argument, that
$$K(a)\ge M\cdot H(s)-O(\log M)$$
\end{enumerate}

Denote $a_i = a_i^1\ldots a_i^M$ for all $i=1\ldots n$ (i.e., we set $a_i$ to be
the row $i$ of the table.) Let us verify that the $n$-tuple of
binary strings $a = (a_1,\ldots,a_n)$ has a complexity profile close to 
the entropy profile of $s$ multiplied by $M$.
\begin{claim}
$\forall U\subseteq[1,\ldots,n], 
K(a_U) \le M\cdot H(s_U) + O(\log M)$
\end{claim}
\begin{proof} 
We extract from the entire table the rows corresponding to $U$; count
frequencies of different columns (of size $|U|$) that occur in this restricted
table (of size $|U|\times M$).  Denote these frequencies  by $f_1,f_2,\ldots $
(of course, the sum of all frequencies equals $1$).  Let $h$ be the entropy of
the distribution with probabilities $f_1,f_2,\ldots $.  By Theorem~5.1 in
\cite{ZL}, 
$$
K(a_U) \le M\cdot h + O(\log M).
$$
Further, we use the fact that frequencies $f_j$ are close to the corresponding
probabilities of $s_u$:
$$
\begin{array}{rclr}
h &=& -\sum_i{f_i\log f_i}
& \text{}\\ 
         &=& -\sum_i{(p_i + O(\frac1{M}))\log(p_i[1+O(\frac1{p_iM})])}
& \text{}\\
         &\le& H(s_U) + O(\frac1{M})  
& \text{}\\
\end{array}
$$
We get the claim by combining the two inequalities.
\end{proof}
\begin{claim}
$\forall U,V\subseteq[1,\ldots,n], 
 K(a_U|a_V) \le M\cdot H(s_U|s_V) + O(\log M)$
\end{claim}
\begin{proof}
Denote $a_V = a^1_V\ldots a^M_V$. We split all positions $i=1\ldots M$ into
classes corresponding to different values  of $a^i_V$. Denote the sizes of
these classes by $m_1,m_1,\ldots$ By property (c) of the table, each $m_j$ must
be proportional to the corresponding probability: the number $m_j$ of positions
$i=1,\ldots,M$ such that $a_V^i=\bar \alpha_j$ is equal to
$$
\mathbf{Prob}[s_v = \bar \alpha_j ] \cdot M + O(1).
$$
Given $a_V$, we  describe $a_U$ by an encoding $a^i_U$ separately for different
classes of positions corresponding to different values of $a^i_V$. Similarly
to the previous Claim, we get
$$
K(a_U|a_V) \le \sum_j{[m_jH(s_U|s_V=\bar \alpha_j) + O(\log m_j)]}
$$
where $m_j$ is the number of columns $c$ of the table where $a^c_V=\bar
\alpha_j$. It follows that 
$$
K(a_U|a_V) \le 
M\sum_j{\frac{m_j}{M}H(s_U|s_V=\alpha_j)} + O(\log M)
= M\cdot H(s_U|s_V) + O(\log M)  
$$
\end{proof}
\begin{claim}
$\forall U,V\subseteq[1,\ldots,n], 
K(a_U|a_V) \ge M\cdot H(\mathbf{s}_U|\mathbf{s}_V) - O(\log|a|)$
\end{claim}
\begin{proof}$$
\begin{array}{rclr}
K(a_U|a_V) &=&   K(a) - K(a_V) - O(\log|a|) 
& \text{by Kolmogorov-Levin Theorem}\\ 
           &\ge& MH(\mathbf{s}) - MH(\mathbf{s}_V) - O(\log|a|) 
& \text{by (c) and previous claim}\\
           &\ge& MH(\mathbf{s}_U|\mathbf{s}_V) - O(\log|a|) 
& \text{Shannon information equality}\\
\end{array}
$$
\end{proof}
Thus, we have constructed a $n$-tuple of binary strings $[a]$ whose complexity
profile is close to $M$ times the entropy profile of $[s]$, up to some
logarithmic term.
\end{proof}

\section{Scaling of secret sharing schemes}
\label{sec:Scaling}

Here, we attempt to show how to scale up and down any secret sharing scheme.
The problem consist of, given a secret sharing for $N$-bit secrets,
constructing new secret sharing schemes for $\ell$-bit secrets where $\ell$ can
be arbitrary large or small. While this task is easy in the perfect case, it
becomes much more difficult in the non-perfect case when we are concerned with
efficiency and information leak.

\subsection{Scaling for perfect schemes}

We present some easy construction for scaling up and down in the perfect case and
state what they achieve in terms of efficiency (size of the shares).

\begin{prop}
\label{prop1}
Let $\Gamma$ be an access structure and $\Sigma$ be a $\PS{\Gamma}{N}{S}$ then

\textup(a\textup) \textup[scaling down\textup] For every positive integer
$\ell\le N$ there exists a \PS{\Gamma}{\ell}{S} 

\textup(b\textup) \textup[scaling up\textup] For every positive integer $q$
there exists a \PS{\Gamma}{qN}{qS} 
\end{prop} 

\begin{proof}

\textup(a\textup) To scale down, we can reuse the same scheme. Simply restrict
the support of the random variable $k$ to $2^\ell$ values and equip this
support with the uniform distribution. Authorized groups can determine the
secret uniquely since it was the case in the initial scheme. Forbidden have no
information about the secret otherwise they had some information in the
initial perfect scheme. 

\textup(b\textup) For scaling up, the new scheme consists of the concatenation
of $q$ independent versions of the initial scheme. Since the new scheme
consists of independent copies (a serialization) of the initial scheme, 
every new entropy value is $q$ times the old entropy value.
\end{proof}

\subsection{Scaling for non-perfect schemes}

Scaling up for nonperfect schemes is similar to the case of perfect
schemes.

\begin{prop}
Let $\Gamma$ be an access structure and $\Sigma$ be a
$\NPS{\Gamma}{N}{S}{\varepsilon}$ then for every non-negative integer $q$ there 
exists a \NPS{\Gamma}{qN}{qS}{q\varepsilon} 
\end{prop}
\begin{proof}
Simply reuse the construction of (b) of proposition~\ref{prop1}. Then a
forbidden group can have at most $q\varepsilon$ bits of information about the
secret.
\end{proof}

Scaling down of the size of the secret becomes non-trivial for non-perfect
secret sharing schemes if we want to keep the same information leak and missing
information. If we can $\varepsilon$-nonperfectly share an $N$-bit secret, then
intuitively it seems that we should be able to share one single bit with
information leak ratio of about $\varepsilon$. However this statement is quite
non-obvious. Now we formulate and prove a slightly weaker statement (it is the
most technical result of this paper):

\begin{theo}
\label{maintheo}
For all $c\in(0,\frac14)$ there exists an integer $N_0>0$ such that for every
access structure $\Gamma$ on $n$ participants. 
If for some $\varepsilon$
there exist a
$\NPS{\Gamma}{N}{S}{\varepsilon}$ where the secret is uniformly distributed,
such that
\begin{itemize}
\item $nS < 2^{cN}$
\item $N > N_0$
\end{itemize}
there exists a  
$\NPS{\Gamma}{1}{S}{\varepsilon'}$ with 
$\varepsilon'=8\varepsilon^{2\over 3}$, where the secret is uniformly distributed
\end{theo}

\emph{Sketch of the proof:} Construct a new scheme for a $1$-bit secret from
the initial scheme in the following way. Given a
$\NPS{\Gamma}{N}{S}{\varepsilon}$ for a uniformly distributed secret in
$\Ks=\{1,\ldots,2^N\}$, take a splitting of $\Ks$ into two equal parts, say
$\Ks_0$ and $\Ks_1$. Then define a new scheme as follows: to share the bit
$i$, take a random element of $\Ks_i$ and share it with the initial scheme. It
is easy to see that this new scheme is indeed a
$\NPS{\Gamma}{1}{S}{\varepsilon'}$ for a uniformly distributed secret bit with
some leak $\varepsilon'$. This leak $\varepsilon'$ depends on the initial
choice of the splitting $\Ks_0$.  We will show that there exists one such
splitting for which the leak is small.

\medskip

We first prove a general lemma about discrete random variables. 

\begin{lem} 
Let $X$ be a finite discrete random variable over a $k$-element set
$A$ (with $k$ even) such that $H(X)\ge\log{k}-\delta$ for some positive
$\delta$. Let $B$ be a random subset of $A$ of size $k/2$ \textup($B$ is
chosen uniformly, i.e., each $(k/2)$-element subset of $A$ is chosen with
probability $1/ {{k}\choose{k/2}}$\textup). 
Then  for every $\gamma\in(0,1)$,  with probability at least 
 $$1-2e^{-\frac{4\tau^2}{k\gamma^2}}$$
 \textup(probability for a random choice of $B$\textup)
we have 
 $$\|\Pr[X\in B] - \frac12\| \le 2\tau$$
 \textup(probability for the initial distribution $X$\textup), where
$\tau=\frac{1+\delta}{2\log{\gamma k }}$.
\label{mainlemma}  
\end{lem}
\noindent
\emph{(In applications of this lemma we will choose the most reasonable values
of parameter $\gamma$.)}

\begin{proof} For each element $x\in A$, denote by $\rho_x$ the non-negative
weight (probability) that $X$ assigns to $x$. Using this notation we have $$
H(X) = \sum_{x\in A}{-\rho_x\log{\rho_x}}$$ A randomly chosen $B$ contains
exactly one half of the points $x$ from $A$. We need to estimate the sum of
$\rho_x$ for all $x\in B$. We do it separately for ``rather large'' $\rho_x$
and for ``rather small'' $\rho_x$. To make this idea more precise, fix a
threshold $\gamma>0$ that separates ``rather large'' and ``rather small''
values of $\rho_x$. Denote by $p_\gamma$ the total measure of all $\rho_x$
that are greater than this threshold. More formally,

  $$p_\gamma  = \sum\limits_{\rho_x>\gamma} \rho_x$$ 
We claim that $p_\gamma$ is rather small. Indeed, if we need to identify some
$x\in A$, we should specify the following information which consists of two
parts:
\begin{itemize}
 \item[1.] We say whether $p_x>\gamma$ or not (one bit of information).
 \item[2a.] If  $p_x>\gamma$, we specify the ordinal number of this ``large''
 point; there are at most $1/\gamma$ points $x'$ such that $\rho_{x'}>\gamma$,
 so we need at most $\log (1/\gamma)$ bits of information;
 \item[2b.] otherwise, $p_x\le \gamma$, we simply specify the ordinal number of
 $x$ in $A$; here we need at most $\log k$ bits of information.
\end{itemize}
From the standard coding argument we get
 $$H(X) \le 1 + p_\gamma  \log(1/\gamma) + (1-p_\gamma) \log k$$
Since $H(X)\ge \log k-\delta$, it follows that 
$
 p_\gamma \le \frac{1+\delta}{\log (\gamma k)}.
$

Thus, we may assume that total measure of ``rather large'' values $\rho_x$ is
quite small even in the entire set $A$; hence, ``large'' points do not affect
seriously the measure of a randomly chosen $B$. It remains to estimate the
typical impact of ``small'' $\rho_x$ to the weight of $B$.

Technically, it is useful to forget about ``large'' points $x$ (substitute
weights $\rho_x>\gamma$ by $0$) and denote
$$
\rho'_x = \left\{
					\begin{array}{ccl}
					\rho_x &\mbox{if } \rho_x\le \gamma\\
					0 &\mbox{otherwise}
					\end{array}
			   \right.
$$
Now we choose exactly $k/2$ different elements from $A$ and estimate the sum of
the corresponding $\rho'_x$. Note that expectation of this sum is one half of
the sum of  $\rho'_x$ for all $x\in A$, i.e, $(1-p_\gamma)/2$. It remains to
estimate the deviation of this sum from its expectation. We use the version of
Hoeffding's bound for samplings without replacement, which can be used to
estimate deviations for a sampling of $k/2$ points from a $k$-elements set,
(\cite{heoffdingineq}[section~6]). The probability of the event that the sum
exceeds expected value plus some $\tau$ can be bounded as follows:
$$
\Pr[\sum\limits_{x\in B} \rho'_x \ge(1-p_\gamma)/2 + \tau]\le e^{-\frac{2\tau^2
}{|B|\gamma^2}}= e^{-\frac{4\tau^2}{k\gamma^2}}
$$
Together with ``large'' values $\rho_x$ we have 
$$
\Pr[\sum\limits_{x\in B} \rho_x \ge(1-p_\gamma)/2 + \tau + p_\gamma]\le
e^{-\frac{4\tau^2}{k\gamma^2}}
$$
Now we fix the parameter $\tau$ to be equal to one half of the upper bound for
$p_\gamma$, i.e., $\tau = \frac{1+\delta}{2\log (\gamma k)}$. It follows that, 
$$
\Pr[\sum\limits_{x\in B} \rho_x \ge 1/2 +2\tau   ]\le 
e^{-\frac{4\tau^2}{k\gamma^2}}
$$
From this bound, we can deduce the symmetric bound for the sum of $\rho_x$ in
$A\setminus B$:
$$
\Pr[\sum\limits_{x\in A\setminus B} \rho_x \le 1/2 -2\tau ]\le 
e^{-\frac{4\tau^2}{k\gamma^2}}
$$
Since $A\setminus B$ and $B$ share the same distribution (the uniform one),
this bound also holds for $B$. Sum up the two bounds and we are done.
\end{proof}

We are now ready to prove Theorem \ref{maintheo}.

\begin{proof} (of Theorem \ref{maintheo}).
Let $\Ks_0$ be a random subset of the set of all secrets $\Ks$ such that
$|\Ks_0|=2^{N-1}$. $\Ks_0$ is chosen uniformly over all possible such fair
splittings of $\Ks$. If $\varkappa$ be the random variable for the $N$-bit
secret in the initial scheme, let us define the new secret bit $\xi$ as the bit
defined by "$\varkappa\in \Ks_0$" ($\xi$ is indeed a bit since $H(\xi) = 1$).
Our goal is to estimate $H(\xi|\sigma_B)$ for any $B\notin\Gamma$ be a
forbidden group, and show it is large.  Formally, we want to show that
$H(\xi|\sigma_B) \ge 1 - \varepsilon'$ where $\varepsilon' =
8\varepsilon^{2\over 3}$. 
 
First, we notice that for any bit $\xi$ constructed as above,
$I(\xi:\sigma_B)\le\varepsilon$ holds for all $B\notin\Gamma$, so we can assume
that $\varepsilon'\le\varepsilon$, i.e,
\begin{equation}
\label{epspgt}
\varepsilon' \ge \frac{8^3}{N^2}
\end{equation}
We know that  $H(\varkappa | \sigma_B)$ is rather large. More precisely,
$$
 H(\varkappa | \sigma_B) \ge N(1 - \varepsilon) 
$$ 
We introduce some positive parameter $\delta$ (to be fixed later) to
separate all values of $\sigma_B$ into two classes:
$$
\mbox{more typical values }b\mbox{ such that } H(\varkappa|\sigma_B=b) \ge
N(1-\delta) 
$$
and 
$$
\mbox{less typical values }b\mbox{ such that }H(\varkappa|\sigma_B=b)< N(1-\delta) 
$$

Since the entropy $H(\varkappa | \sigma_B)$ is large, the total measure of all
``less typical'' values $b$ is rather small (more precisely, it is not greater
than $\frac{\varepsilon}{\delta}$). We do not care about the conditional
entropy of $\xi$ when $b$ is non-typical (the total weight of these $b$ is so
small that they do not contribute essentially to $H(\xi|\sigma_B)$). We focus
on the contribution of $H(\xi|\sigma_B=b)$ for a typical value $b$. To estimate
this quantity we apply lemma \ref{mainlemma} to the distribution $k$
conditional to $\sigma_B=b$, it follows that 
$$
H(\xi|\sigma_B=b) \ge h(1/2 + 2\tau) \ge 1 - 16\tau^2
\mbox{ with probability } 1-2e^{-\frac{4\tau^2}{\gamma^2}2^{-N}}
$$
for some new parameter $\gamma >0$ and $\tau
=\frac{1+\delta N}{2(\log{\gamma}+N)}$.

This inequality true for all forbidden group $B$ and any typical share $b$.
Thus if we sum up the bad events, we obtain that the following estimation for
$H(\xi|\sigma_B)$:
\begin{eqnarray*}
H(\xi|\sigma_B) & = & \sum_{b\in\Ss_B}{\Pr[\sigma_B=b]H(\xi|\sigma_B=b)} \\
         &\ge& \sum_{\text{typical } b}{\Pr[\sigma_B=b]H(\xi|\sigma_B=b)}\\ 
         &\ge& (1-{\varepsilon\over\delta})(1-16\tau^2)  \\
         &\ge&1-{\varepsilon\over\delta} - 16\tau^2 \\      
\end{eqnarray*}
holds with probability at least 
\begin{equation}
\label{eqproba}
1-|\overline{\Gamma}||\Ss_\Ps|2e^{-\frac{4\tau^2}{\gamma^2}2^{-N}}
\end{equation}
where $\Ss_\Ps$ is the set of all possible shares given to the group of all
participants.

Now, we choose our parameters $\gamma$ and $\delta$ to deduce our result
and show that our choice is valid. We take 
\begin{eqnarray}
\label{paramtau}
16\tau^2 =
{\varepsilon\over\delta} =
 \frac12\varepsilon'
= 4\varepsilon^{2\over 3} 
\end{eqnarray}
Under these conditions 
\begin{eqnarray}\label{paramgamma}
\log\gamma=-N\left[1-\frac18\left(\frac{\varepsilon'}{\varepsilon N}+ 2\right)\right]
\end{eqnarray}
and
$$
H(\xi|\mathbf{B})\ge 1-8\varepsilon^{2\over 3} = 1-\varepsilon'
$$

We want to find a simple sufficient condition that guarantees  that the
probability (\ref{eqproba}) is non-negative. To this end we do some (rather
boring) calculations. We take the required inequality and reduce it step by
step to a weaker but more suitable form:
$$
\begin{array}{rcl|l}
|\overline{\Gamma}||S_\Ps|2e^{-\frac{4\tau^2}{\gamma^2}2^{-N}} &<& 1  
& \text{ that is what we need, see }(\ref{eqproba}) \\
|\overline{\Gamma}| |S_\Ps|  &<& 2e^{\frac{4\tau^2}{\gamma^2}2^{-N}} \\
                   2^n2^{nS} &<& 2e^{\frac{4\tau^2}{\gamma^2}2^{-N}} 
& \text{trivial upper bounds for } |\overline{\Gamma}| \text{ and } |S_\Ps|\\
                  2^{n(S+1)} &<& 2^{\frac{4\tau^2}{\gamma^2}2^{-N}} 
& \text{since } e > 2 \\
              n(S+1)  &<& \frac{4\tau^2}{\gamma^2}2^{-N} 
& \text{by applying }\log\\
                  2nS &<& \frac{4\tau^2}{\gamma^2}2^{-N} 
& \text{since } S \ge 1 \\
                  2nS &<& \frac{\varepsilon'}{8}
2^{N(1-\frac14(\frac{\varepsilon'}{\varepsilon N}+ 2))}
& \text{from } (\ref{paramtau}) \text{ and } (\ref{paramgamma}) \\
                   nS &<& \varepsilon'2^{\frac14N-4}
& \text{since } \varepsilon'\le\varepsilon N  \\
               2^{cN} &<& \varepsilon'2^{\frac14N-4}
& \text{by assumption} \\
                    1 &<& \varepsilon'2^{(\frac14-c)N-4}\\
                    0 &<& (\frac14-c)N+\log{\varepsilon'}-4\\
                    0 &<& (\frac14-c)N-2\log{N}+5
&\text{from } (\ref{epspgt})  
\end{array} 
$$
The last inequality (which is a sufficient condition for (\ref{eqproba}) to be
non-negative) holds when $c<\frac14$ and $N>N_0$ for some large enough $N_0$
depending on $c$. 
\end{proof}

Notice that in this case we consider schemes where the secret is uniformly
distributed since the dependency on the probability distribution of the secret
is not trivial in the nonperfect case. Sharing exactly one bit instead of $N$
seems more difficult. We do not know whether this bound can be improved, in
particular, can we achieve a leak of $O(\varepsilon)$ ? The assumption $nS =
O(2^N)$ points out that the result holds for various kind of access structures
defined by some trade-off between the number of participants $n$ and the size
of the shares $S$ of a scheme for $N$-bit secrets.

\section{Conclusion}

In this article we introduced several definitions of almost-perfect secret
sharing schemes (two versions in the framework of Shannon's entropy and another
version in the framework of Kolmogorov complexity). We proved that all these
approaches are asymptotically equivalent (have equivalent asymptotical rates of
schemes for each access structure). This means that we can combine tools of
Shannon's information theory and Kolmogorov complexity to investigate the
properties of approximately-perfect secret sharing.

The major questions remain open. The most important one is to understand: can
almost perfect secret sharing schemes achieve substantially
better information rates than perfect (in classic sense) secret sharing schemes?
The known proofs of lower bounds for the rate of perfect secret sharing
schemes are based on combinations of information inequalities; so it is not
hard to check that the same type of arguments imply the same kind of bounds for
almost perfect schemes. Thus, the problem of separating the information rates
for \emph{almost-perfect} and exactly \emph{perfect} schemes looks rather hard.

\section*{Acknowledgment}

The author would like to thank Andrei Romashchenko and Sasha Shen for
stimulating discussions, and anonymous reviewers who helped substantially
improve the manuscript. This work is partially supported by EMC
ANR-09-BLAN-0164-01 and NAFIT ANR-08-EMER-008-01 grants. 

\bibliography{../biblio}

\end{document}